\documentclass[onecolumn,10pt]{article}
\usepackage{amsmath}
\usepackage{amsfonts}
\usepackage{amssymb}
\usepackage{amsthm}
\usepackage{times}
\usepackage{cancel}
\usepackage{graphics}

\newcommand{\trace}[1]{\mbox{$\mathrm{Tr}$}(#1)}

\newcommand{\hil}[1]{\mbox{$\mathcal{#1}$}}

\newcommand{\ket}[1]{| #1 \rangle}
\newcommand{\bra}[1]{\langle #1 |}
\newtheorem*{infolossthm}{The Information Loss Theorem}

\newtheorem*{corollary}{Corollary}

\begin{document}

\date{}
\title{\textbf{Two Dogmas About Quantum Mechanics}}
\author{\textbf{Jeffrey Bub} \\
{\footnotesize Department of Philosophy, University of Maryland, College
Park, MD 20742\thanks{\textit{E-mail address:} jbub@umd.edu} }\\
\textbf{Itamar Pitowsky}\\
{\footnotesize Department of Philosophy, The Hebrew University, Jerusalem,
Israel\thanks{\textit{E-mail address:} itamarp@vms.huji.ac.il}}}
\maketitle

\begin{abstract}
We argue that the intractable part of the measurement problem---the `big' measurement problem---is a pseudo-problem that depends for its legitimacy on the acceptance of two dogmas. The first dogma is John Bell's assertion that measurement should never be introduced as a primitive process in a fundamental mechanical theory like classical or quantum mechanics, but should always be open to a complete analysis, in principle, of how the individual outcomes come about dynamically. The second dogma is the view that the quantum state has an ontological significance analogous to the significance of the classical state as the `truthmaker' for propositions about  the occurrence and non-occurrence of events, i.e., that the quantum state is a  representation of physical reality. We show how both dogmas can be rejected in a realist information-theoretic interpretation of quantum mechanics as an alternative to the Everett interpretation. The Everettian, too, regards the `big' measurement problem as a pseudo-problem, because the Everettian rejects the assumption that measurements have definite outcomes, in the sense that one particular outcome, as opposed to other possible outcomes, actually occurs in a quantum measurement process. By contrast with the Everettians, we accept that measurements have definite outcomes. By contrast with the Bohmians and the GRW `collapse' theorists who add structure to the theory and propose dynamical solutions to the `big' measurement problem, we take the problem to arise from the failure to see the significance of Hilbert space as a new \emph{kinematic} framework for the physics of an indeterministic universe, in the sense that Hilbert space imposes kinematic (i.e., pre-dynamic) objective probabilistic constraints on correlations between events. 
\end{abstract}

\section{Oxford Everett}

The salient difference between classical and quantum mechanics is the noncommutativity of the operators representing the physical magnitudes (`observables') of a quantum mechanical system---or, equivalently, the transition from a classical event space, represented by the Boolean algebra of (Borel) subsets of a phase space, to a non-Boolean quantum event space represented by the projective geometry of closed subspaces of a Hilbert space, which form an infinite collection of intertwined Boolean algebras, each Boolean algebra corresponding to a resolution of the identity: a partition of the Hilbert space representing a family of mutually exclusive and collectively exhaustive events.

Probabilities in quantum mechanics are, as von Neumann put it \cite[p. 245]{Neumann1954}, `uniquely given from the start' as a nonclassical relation
between events represented by the angles between the 1-dimensional subspaces representing atomic (elementary) events in the projective geometry of
subspaces of Hilbert space. If $e$ and $f$ are atomic events, the  `transition probability' (Born probability) between the events is: 
\begin{equation}
\mbox{prob}(e,f)  = |\langle e|f\rangle|^{2} =  |\langle f|e\rangle|^{2} = \cos^{2}\theta_{ef}
\end{equation}
The transition probability 
can be expressed as:
\begin{equation}
\mbox{prob}_{e}(f) = \trace{P_{e}P_{f}}
\end{equation}
where $P_{e}$ and  $P_{f}$  are the projection operators onto the 1-dimensional subspaces representing the events $e$ and $f$, respectively. Uniqueness is shown by Gleason's theorem \cite{Gleason}:\footnote{For von Neumann, uniqueness is a consequence of invariance under the unitary symmetries of the projective lattice representing events.} in a Hilbert space $\hil{H}$ of dimension greater than 2, if $\sum_{i}\mbox{prob}(f_{i}) = 1$ for the atomic events $f_{i}$ in each Boolean algebra generated by a partition of the Hilbert space into orthogonal 1-dimensional subspaces, then the probabilities of events $f$ represented by subspaces of $\hil{H}$ are uniquely represented as:
\begin{equation}
\mbox{prob}_{\rho}(f) = \trace{\rho P_{f}}
\end{equation}
where $P_{f}$ is the projection operator onto the subspace representing the event $f$ and $\rho$ is a density operator representing a pure state ($\rho = P_{e}$, for some atomic event $e$) or a mixed state ($\rho = \sum_{i}w_{i}P_{e_{i}}$).  

It is assumed that
the assignment of probabilities satisfies a condition that Barnum et al \cite
{BCFFS2000} call `the noncontextuality of probability,' that the probability
assigned to an event $f$ depends only on $f$ and is independent of the Boolean algebra to which the event belongs. 
Note that if `$f$
 in context 1' and `$f$ in context 2'
represented two distinct events, we could not represent the structure of quantum events as the projective geometry of subspaces of a Hilbert space: we would have to enlarge the structure. 

The question is: what do these `transition probabilities' or `transition weights' mean? The probabilities are probabilities of---\emph{what}? Evidently, $|\langle e|f\rangle|^{2}$ does not represent the probability of a spontaneous transition from an event $e$ to the event $f$. The textbook answer is that $|\langle e|f\rangle|^{2}$ represents the probability, for a system in the state $\ket{e}$  in which the event $e$ has probability 1, of finding the event $f$ in a measurement of an observable of the system, where the set of possible outcomes of the measurement generates a Boolean algebra, representing a partition of the Hilbert space containing the event $f$ (but note, not the event $e$). 

The textbook answer by itself, without adding anything more to the story of how these events are supposed to come about in a measurement process, is adequate only if we are content with an instrumentalist interpretation of the theory. Why? The structure of the quantum event space determines the kinematic part of quantum theory. This includes the association of  Hermitian
operators with observables, the Born probabilities, the von Neumann-L\"{u}ders 
conditionalization rule, and the unitarity constraint on the dynamics, which is related to the event structure via a theorem of Wigner \cite{Wigner1959},\cite{Uhlhorn1963}. The transition from the state 
$\ket{e}$, in which the event $e$ has probability 1, to the state $\ket{f}$, in which the event $f$ has probability 1, with probability $|\langle e|f\rangle|^{2}$ in a measurement process is a non-unitary stochastic transition that is not described by the unitary dynamics. Since the probability of the event $e$ was 1 before the measurement and is now, in the state $\ket{f}$  after the occurrence of the measurement outcome $f$, less than 1, there is a loss of information on measurement or---as Bohr put it---an `irreducible and uncontrollable' measurement disturbance. Without a dynamical explanation of this measurement disturbance, or an analysis of what is involved in a quantum measurement process that addresses the issue (including, possibly, rejecting the `eigenvalue-eigenstate rule'---the association of the outcome event $f$ with the state $\ket{f}$---as in Bohm's theory or modal interpretations), the theory qualifies  as an algorithm for predicting the probabilities  of measurement outcomes, but cannot be regarded as providing a realist account, in principle, of how events come about in a measurement process. 

This is the measurement problem. Proposed solutions to the problem, such as Bohm's `hidden variable' theory \cite{GoldsteinSEP} or the GRW `dynamical collapse' theory \cite{GhirardiSEP}, add structure to the theory: particle trajectories in the case of Bohm's theory or a non-unitary stochastic dynamics for the `primitive ontology' of the GRW theory: mass density in the GRWm version, or `flashes' in the GRWf version (see, e.g., /cite{Allori2007}). The Everett interpretation purports to solve the problem without adding any new structural elements to quantum mechanics.       

The central claims of the Everett
interpretation in the `Oxford' version developed by Deutsch \cite{Deutsch1999},
Saunders \cite{Saunders1995,Saunders1998,Saunders2004}, Wallace \cite
{Wallace2003a,Wallace2003b,Wallace2003c,Wallace2005}, Greaves and others \cite
{Greaves2004,Greaves2007a,Greaves2007b} can be outlined as follows:

\begin{description}
\item[Ontology] At the most fundamental level, what there \emph{is} is
described by the quantum state of the universe---so whatever is true or
false is determined by the quantum state as the `truthmaker' for propositions about  the occurrence and non-occurrence of events.
\item[Branching] A family of effectively non-interfering or decoherent histories of coarse-grained events associated with relatively stable systems at the macrolevel emerges through the dynamical process of decoherence, as a consequence of the Hamiltonian that characterizes the dynamical evolution of the universal quantum state. With respect to the coarse-grained basis selected by decoherence, the quantum state decomposes into a linear superposition that can be interpreted as describing an emergent branching structure of non-interfering quasi-classical histories or `worlds,'  identified with the familiar classical macroworlds of our experience, weighted by the Born probabilities. The alternative outcomes of a quantum measurement process are associated with different branches in the decomposition of the quantum state with respect to the decoherence basis. There is no fact of the matter as to the number of branches: the history space is a quasi-classical probability space that is inherently vaguely defined (appropriately so, given the vague specification of macro-configurations). The coarse-graining of the event space can be refined or coarsened to a certain extent without compromising effective decoherence, and the decoherence basis can be unitarily transformed (e.g., rotated) over a certain range of transformations without compromising decoherence. 
\item[Uncertainty/Caring] There is a sense in which a rational agent on a branch,
faced with subsequent branching, can be uncertain about the future (i.e.,
uncertain about `which branch the agent will subsequently occupy'). Such an
agent can have rational credences (degrees of belief that satisfy the axioms
of probability theory) about the outcomes of quantum measurements, even
though all outcomes occur on different branches. Alternatively, even without uncertainty, an agent faced with multiple futures will care about what happens on a branch, and so will have a `caring measure' for decision-making that quantifies the extent of caring for different branches and satisfies the axioms of probability theory.
\item[Probability] To achieve a realist interpretation of quantum mechanics
that solves the measurement problem, it suffices to postulate that an agent's credence function or caring measure conforms to the objective quantum mechanical weights of the
different branches. In fact, it is possible to prove that this must be so,
given standard rationality constraints on an agent's preferences, and a measurement neutrality 
assumption: that a rational agent is indifferent between two quantum wagers that agree on the quantum state, the observable measured, and the payoff function on the outcomes, i.e., the agent is indifferent between alternative measurement procedures; alternatively, the result follows from a related equivalence assumption: that a rational agent assigns equal credences to events that are assigned equal quantum weights. These additional assumptions can be justified as 
rationality constraints, but only on the Everett interpretation, in which
all possible measurement outcomes occur, relative to different branches.
\end{description}

The Everettian aims to show that standard quantum mechanics can be
understood as a complete theory in a realist sense---that the measurement problem does not
reduce the theory to an instrument for the probabilistic prediction of
measurement outcomes. The basic problem for the Everettian is to `save the
appearances,' given the radical difference between our experience of a
stable macroworld and the ontological assumption. The dynamics of decoherence yields
an emergent weighted branching structure of quasi-classical histories at the macrolevel. So what has to be explained is
how uncertainty or caring makes sense when all alternatives occur relative to different branches, and how the quantum weights---which are a
feature of the quantum state, i.e., the ontology---are associated with the
credence function or caring measure of rational agents. The measurement problem is the problem of
explaining the apparently `irreducible and uncontrollable disturbance' in a
quantum measurement process, the `collapse' of the wave function described
by von Neumann's projection postulate. The Everettian's solution is to show
how appearances can be saved by denying that there is any such disturbance,
on the basis that no definite outcome is selected in a measurement---all
outcomes are selected relative to different branches, according to the quantum
theory. The appearance of disturbance on a single branch is a reflection of
how the quantum weights are distributed in the emergent process of
branching, and if we either assume or prove that our credence function or caring measure should
conform to these weights, then we have an explanation for the appearance of
disturbance in a realist interpretation of quantum mechanics as a complete
dynamical theory.

Of course, everything hinges on whether the different components of the
intepretation can be established satisfactorily, and there is now an
extensive literature challenging and defending these claims, especially 
\textbf{Uncertainty/Caring} and \textbf{Probability}. Here we simply list these
components\footnote{For a critique of \textbf{Probability} by one of us, see Hemmo and Pitowsky \cite
{HemmoPitowsky2007}.} and note that the claim is that the Everett
interpretation solves the measurement problem on the basis of (i) the
weighted branching structure of quasi-classical histories that
emerges through the dynamical process of decoherence, (ii) an argument that rational agents can be
uncertain or care differently about different  futures in a branching universe, and (iii) the proposal that the credence function or caring measure
of rational agents should conform to the weights of the branches. For the
Everettian, the icing on the cake is that the interpretation yields a
derivation of Lewis's Principal Principle: the identification of an
objective feature of the world---the quantum weights---with the credence function or caring measure of
rational agents, and hence the interpretation of the quantum weights as
objective chances. But the cake itself, so to speak, is independent of
this additional feature. (See Wallace \cite{Wallace2005}.)
 
In a previous publication \cite{Pitowsky07}, one of us characterized debates about  the foundations of quantum mechanics in terms of two assumptions or dogmas, and distinguished two measurement problems: a `big' measurement problem and a `small' measurement problem. The first dogma is Bell's assertion (defended in \cite{Bellmeas}) that measurement should never be introduced as a primitive process in a fundamental mechanical theory like classical or quantum mechanics, but should always be open to a complete analysis, in principle, of how the individual outcomes come about dynamically. The second dogma is the view that the quantum state has an ontological significance analogous to the ontological significance of the classical state as the `truthmaker' for propositions about the occurrence and non-occurrence of events, i.e., that the quantum state is a  representation of physical reality. The `big' measurement problem is the problem of explaining how measurements can have definite outcomes, given the unitary dynamics of the theory: it is the problem of explaining \emph{how individual measurement outcomes come about dynamically.} The `small' measurement problem is the problem of accounting for our familiar experience of a classical or Boolean macroworld, given the non-Boolean character of the underlying quantum event space: it is the problem of explaining  the \emph{dynamical emergence of an effectively classical  probability space of macroscopic measurement outcomes} in a quantum measurement process. 

The `big' measurement problem depends for its legitimacy on the acceptance of the two dogmas. We argue below that both dogmas should be rejected, and that the `big' measurement problem is a pseudo-problem. In a sense, the Everettian, too, regards the `big' measurement problem as a pseudo-problem, because the Everettian rejects the assumption that measurements have definite outcomes, in the sense that one particular outcome, as opposed to other possible outcomes, actually occurs in a quantum measurement process. By contrast with the Everettians, we accept that measurements have definite outcomes. By contrast with the Bohmians and the GRW `collapse' theorists who add structure to the theory and propose dynamical solutions to the `big' measurement problem, we take the problem to arise from the failure to see the significance of Hilbert space as a new \emph{kinematic} framework for the physics of an indeterministic universe, in the sense that Hilbert space imposes kinematic (i.e., pre-dynamic) objective probabilistic constraints on correlations between events. By `predynamic' here, we refer to generic features of quantum systems, independent of the details of the dynamics (see Jannsen \cite{Jannsen2007a} for a similar kinematic-dynamic distinction in the context of special relativity). The `small' measurement problem is resolved by considering the dynamics of the measurement process and the role of decoherence in the emergence of an effectively classical probability space of macro-events to which the Born probabilities refer (alternatively, by considering certain combinatorial features of the probabilistic structure: see Pitowsky \cite[\S4.3]{Pitowsky07}).

In the following section, we list the essential features of the proposed information-theoretic interpretation, somewhat more extensively than our brief sketch of the Everett interpretation. Further discussion follows in a subsequent Commentary.

\section{An Information-Theoretic Interpretation of Quantum Mechanics}

The elements of the
information-theoretic interpretation we propose\footnote{For related views, see Demopoulos \cite{Demopoulos2008}, Pitowsky \cite{PitowskyBetting,Pitowsky07}.} can be set out as
follows:

\begin{description}
\item[`No Cloning'] The empirical discovery underlying the transition from classical to quantum mechanics is the discovery that chance set-ups behave differently than we thought they did. More precisely: there are information sources that cannot be broadcast---there is no universal cloning machine capable of copying the outputs of an arbitrary information source.
\item[Kinematics] Hilbert space as a projective geometry
(i.e., the subspace structure of Hilbert space) represents a non-Boolean event space, in which there are built-in, structural probabilistic constraints on correlations between events (associated with the angles between events)---just as in special
relativity the geometry of Minkowski space-time represents spatio-temporal
constraints on events. Certain principles characterizing physical processes motivate the choice of Hilbert space as the representation space for the correlational structure of events, just as Einstein's
principle of special relativity and the light postulate motivate the choice of Minkowski space-time as the representation space for the spatio-temporal structure of events. In the case of
quantum mechanics, these principles are information-theoretic and include a `no signaling' principle and a `no cloning' principle. The structure of Hilbert space imposes kinematic (i.e., pre-dynamic) objective probabilistic constraints
on events to which a quantum dynamics of matter and fields is required to conform, through its symmetries, just
as the structure of  Minkowski space-time imposes kinematic constraints on events to which a
relativistic dynamics is required to conform. In this sense \emph{Hilbert space provides the kinematic framework for the physics of an indeterministic universe}, just as Minkowski space-time provides the kinematic framework for the physics of  a non-Newtonian, relativistic universe. There is no deeper explanation for the quantum phenomena of interference and entanglement than that provided by the structure of Hilbert space, just as there is no deeper explanation for the relativistic phenomena of Lorentz contraction and time dilation than that provided by the structure of Minkowski space-time.
\item[Dynamics] The unitary quantum dynamics evolves the whole structure of events with probabilistic correlations in Hilbert space (in the Heisenberg picture), not the evolution from one configuration of the universe to
another, i.e., not the evolution from one actual co-occurrence of events to a subsequent actual co-occurrence of events. This means that there can be a real change in the correlations between events at the microlevel without a change in the occurrence of events at the macrolevel (as in the evolution of a quantum system through the unitary gates of a quantum computer, prior to the final measurement). 
\item[Probability] By Gleason's theorem, there is a unique assignment of credences conforming to the structural probabilistic constraints (the objective chances) of Hilbert space (see Pitowsky \cite{PitowskyBetting}). These credences are encoded in the quantum state. So the quantum state is a credence function. 
\item[Information Loss] The salient
principle marking the transition from classical to nonclassical theories 
of information is the `no cloning' principle: there is no universal cloning machine capable of copying the outputs of an arbitrary information source.\footnote{More precisely, there is no universal broadcasting machine. See below.} This principle entails a loss of information in a measurement
process---an `irreducible and uncontrollable disturbance'---\emph{irrespective of
how the measurement process is implemented dynamically}. The loss of
information is to be understood, ultimately, as a kinematic  effect of the nonclassical quantum event space, just
as Lorentz contraction is, ultimately, a kinematic effect in special
relativity. 
\item[Completeness] Conditionalizing on a measurement outcome leads to a
nonclassical updating of the credence function represented by the quantum state via the von Neumann-L\"{u}ders rule, which expresses the information loss on measurement. This updating is consistent with a dynamical account of
the correlations between micro and macro-events in a quantum measurement process. The Hamiltonians characterizing the interactions between microsystems and macrosystems, and the interactions between macrosystems and their environment, are such that certain relatively stable structures of events associated with the familiar macrosystems of our experience emerge at the macrolevel, forming an effectively classical probability space. This amounts to a consistency proof that, say, a Stern-Gerlach spin-measuring device or a bubble chamber behaves dynamically according to the kinematic constraints represented by the projective geometry of Hilbert space, as these constraints manifest themselves at the macrolevel. Such a consistency proof demonstrates the completeness of quantum
mechanics. Given the `no cloning' principle underlying the kinematics of Hilbert space, there is no further story to be told about how individual measurement outcomes come about dynamically (assuming we don't add structure to the theory, such as Bohmian trajectories or dynamical `collapses'). Similarly, the dynamical explanation of relativistic phenomena like Lorentz contraction in terms of forces, insofar as the forces are required to be Lorentz covariant, amounts to a consistency proof. There is no further story to be told about Lorentz contraction, once it is shown how to
provide a dynamical account   consistent with the kinematic constraints of
Minkowski geometry (assuming we don't add structure to the theory, such as the ether). 
\item[Realism] The possibility of a dynamical analysis of measurement processes consistent with the
Hilbert space kinematic constraints justifies the information-theoretic
interpretation of quantum mechanics as realist and not merely a predictive
instrument for updating probabilities on measurement outcomes.
\end{description}

\section{Commentary}

On the information-theoretic interpretation, the quantum state is a credence function, a bookkeeping device for keeping track of probabilities---the
universe's objective chances---not the quantum analogue of the dynamically evolving
classical state understood as the `truthmaker' for propositions about  the occurrence and non-occurrence of events.

Conditionalization on the occurrence of an event $a$, in the sense of a
minimal revision---consistent with the subspace structure of Hilbert
space---of the probabilistic information encoded in a quantum state given by a density operator $\rho$, is given by the von Neumann-L\"{u}ders rule:\footnote{
See \cite{BubProjPost} for a discussion.} 
\begin{equation}
\rho \rightarrow \rho_{a} \equiv \frac{P_{a}\rho P_{a}}{\mbox{$\mathrm{Tr}$}
(P_{a} \rho P_{a})}  \label{eqn:luders}
\end{equation}
where $P_{a}$ is the projection operator onto the subspace representing the
event $a$. That is, $\rho_{a}$ is the conditionalized density operator,
conditional on the event $a$, and the normalizing factor $
\mbox{$\mathrm{Tr}$}(P_{a} \rho P_{a}) = \mbox{$\mathrm{Tr}$}(\rho P_{a})$
is the probability assigned to the event $a$ by the state $\rho$.

If we consider a pair of correlated systems, A and B, then
conditionalization on an A-event, for the probabilistic information encoded
in the density operator $\rho_{B}$ representing the probabilities of events
at the remote system B, will always be an updating, in the sense of a
refinement.

For example, suppose the system A is associated with a 3-dimensional Hilbert
space $\mbox{$\mathcal{H}$}_{A}$ and the system B is associated with a
2-dimensional Hilbert space $\mbox{$\mathcal{H}$}_{B}$. Suppose the
composite system AB is in an entangled state: 
\begin{eqnarray}
| \psi^{AB} \rangle & = & \frac{1}{\sqrt{3}}(| a_{1} \rangle| b_{1} \rangle
+ | a_{2} \rangle| c \rangle + | a_{3} \rangle| d \rangle)  \notag \\
& = & \frac{1}{\sqrt{3}}(| a^{\prime}_{1} \rangle| b_{2} \rangle + |
a^{\prime}_{2} \rangle| e \rangle + | a^{\prime}_{3} \rangle| f \rangle)
\end{eqnarray}
where $| a_{1} \rangle,| a_{2} \rangle,| a_{3} \rangle$ and $|
a^{\prime}_{1} \rangle,| a^{\prime}_{2} \rangle,| a^{\prime}_{3} \rangle$
are two orthonormal bases in $\mbox{$\mathcal{H}$}_{A}$ and $| b_{1}
\rangle, | b_{2} \rangle$ is an orthonormal basis in $\mbox{$\mathcal{H}$}
_{B}$. The triple $| b_{1} \rangle,| c \rangle,| d \rangle$ and the triple $
| b_{2} \rangle,| e \rangle,| f \rangle$ are nonorthogonal triples of
vectors in $\mbox{$\mathcal{H}$}_{B}$.\footnote{
The vectors in each triple are separated by an angle $2\pi/3$. For a precise
specification of these vectors, see Bub \cite{Bub2007}.}The state of B (obtained
by tracing over $\mbox{$\mathcal{H}$}_{A}$) is the completely mixed state $
\rho_{B} = \frac{1}{2}I_{B}$: 
\begin{equation}
\frac{1}{3}| b_{1} \rangle\langle b_{1} | + \frac{1}{3}| c \rangle\langle c
| + \frac{1}{3}| d \rangle\langle d | = \frac{1}{3}| b_{2} \rangle\langle
b_{2} | + \frac{1}{3}| e \rangle\langle e | + \frac{1}{3}| f \rangle\langle
f | = \frac{I_{B}}{2}
\end{equation}

Conditionalizing on one of the eigenvalues $a_{1},a_{2},a_{3}$ or $
a_{1}^{\prime },a_{2}^{\prime },a_{3}^{\prime }$ of an A-observable $A$ or $
A^{\prime }$ via (\ref{eqn:luders}), i.e., on the occurrence of an event
corresponding to $A$ taking the value $a_{i}$ or $A^{\prime }$ taking the
value $a_{i}^{\prime }$ for some $i$, changes the density operator $\rho _{B}
$ of the remote system B to one of the states $|b_{1}\rangle ,|c\rangle
,|d\rangle $ or to one of the states $|b_{2}\rangle ,|e\rangle ,|f\rangle $.
Since the mixed state $\rho _{B}=\frac{1}{2}I_{B}$ can be decomposed as an
equal weight mixture of $|b_{1}\rangle ,|c\rangle ,|d\rangle $ and as an
equal weight mixture of $|b_{2}\rangle ,|e\rangle ,|f\rangle $, the change
in the state of B is an updating, in the sense of a refinement of the
information about B encoded in the state $|\psi ^{AB}\rangle $, taking into
account the new information $a_{i}$ or $a_{i}^{\prime }$. In fact, the mixed
state $\rho _{B}=\frac{1}{2}I_{B}$ corresponds to an infinite variety of
mixtures of pure states in $\mbox{$\mathcal{H}$}_{B}$ (not necessarily equal
weight mixtures, of course). The effect at the remote system B of
conditionalization on any event at A will always be an updating, in the
sense of a refinement, with respect to one the these mixtures.\footnote{
Fuchs makes a similar point in \cite{FuchsInfo4}.} This is the content of
the Hughston-Jozsa-Wootters theorem \cite{HJW}. It is what Schr\"{o}dinger
called `remote steering' and is the basis of quantum teleportation, quantum
dense coding, and other peculiarities of quantum information, including the
impossibility of unconditionally secure bit commitment (see Bub \cite{BubQIC}
for a discussion).

The effect of conditionalization at a remote system (the system that is not
directly involved in the conditionalizing event) is then consistent with a
`no signaling' principle: 
\begin{eqnarray}
\sum_{b}p(ab|AB) \equiv p(a|AB) & = & p(a|A) \\
\sum_{a}p(ab|AB) \equiv p(b|AB) & = & p(b|B) \ 
\end{eqnarray}
where $a$ represents a value of $A$ and $b$ represents a value of $B$. If
conditionalization on the value of an A-observable changed the probabilities
at a remote system B in a way that could \emph{not} be represented as an
updating in the sense of a refinement of the prior information about B
expressed in terms of correlations between A-observables and B-observables
(as encoded in the entangled state $| \psi^{AB} \rangle$), then
conditionalization would allow instantaneous signaling between A and B. The
occurrence of a particular sort of event at A---corresponding to a determinate value
for the observable $A$ as opposed to a determinate value for some other observable $A^{\prime}$---would produce a detectable change in the B-probabilities, and so Alice at
A could signal instantaneously to Bob at B merely by performing an $A$
-measurement and gaining a specific sort of information about A (the value
of $A$ or the value of $A^{\prime}$).

The `no signaling' principle is a special case of what Barnum et al \cite
{BCFFS2000} call `the noncontextuality of probability,' which can be
expressed as a condition on the probabilities assigned to the eigenvalues of
any two commuting observables $[X,Y] = 0$: 
\begin{eqnarray}
\sum_{y}p(xy|XY) \equiv p(x|XY) & = & p(x|X) \\
\sum_{x}p(xy|XY) \equiv p(y|XY) & = & p(y|Y) \ 
\end{eqnarray}
This formulation of the noncontextuality of probability follows from the
representation of an observable in terms of its spectral measure.\footnote{
Barnum et al formulate noncontextuality as the requirement that the
probability assigned to an event $e$ depends only on $e$ and is independent
of the other events in each mutually exclusive and collectively exhaustive
set of events $\{e_{i}\}$ containing $e$, i.e., that the probability of an
event is independent of the Boolean subalgebra to which the event belongs.}
We obtain the `no signaling' condition if we take $X = A\otimes I$ and $Y =
I \otimes B$. Note that `no signaling' is not specifically a relativistic
constraint on superluminal signaling. It is simply a condition imposed on
the marginal probabilities of events for separated systems, requiring that
the marginal probability of a B-event is independent of the particular set
of mutually exclusive and collectively exhaustive events selected at A, and
conversely, and this might well be considered partly constitutive of what
one means by separated systems.

To preserve the `no signaling' principle, quantum probabilities must also
satisfy a `no cloning' principle: there can be no universal cloning machine,
i.e., it is impossible to construct a cloning machine that will clone the
output of an arbitrary information source. More precisely, there can be no
universal broadcasting machine---no device that takes a probability
distribution over an event space to a new probability distribution over a
product space of events, where the marginal probability distributions over
each factor space is the same as the original distribution. We will continue
to use the term `cloning' rather than `broadcasting' because it is more
intuitive and more familiar, but note that we have in mind copying the \emph{outputs} of an information source, not the information source itself
(defined by the probability distribution).

Suppose a universal cloning machine were possible. Then such a device could
copy any state in the orthogonal triple $| b_{1} \rangle, | c \rangle, | d
\rangle$ as well as any state in the orthogonal triple $| b_{2} \rangle, | e
\rangle, | f \rangle$. It would then be possible for Alice at A to signal to
Bob at B. If Alice obtains the information given by an eigenvalue $a_{i}$ of 
$A$ or $a^{\prime}_{i}$ of $A^{\prime}$, and Bob inputs the system B into
the cloning device $n$ times, he will obtain one of the states $| b_{1}
\rangle^{\otimes n}, | c \rangle^{\otimes n}, | d \rangle^{\otimes n}$ or
one of the states $| b_{2} \rangle^{\otimes n}, | e \rangle^{\otimes n}, | f
\rangle^{\otimes n}$, depending on the nature of Alice's information. Since
these states tend to mutual orthogonality in $\otimes^{n}\mbox{$%
\mathcal{H_{B}}$}$ as $n \rightarrow \infty$, they are distinguishable in
the limit. So, even for finite $n$, Bob would in principle be able to obtain
some information instantaneously about a remote event.

More fundamentally, the existence of a universal cloning machine is
inconsistent with the interpretation of Hilbert space as providing the kinematic framework for an indeterministic physics, in which probabilities (objective chances)  are `uniquely
given from the start' by the geometry of Hilbert space. For such a device
would be able to distinguish the equivalent mixtures of nonorthogonal states
represented by the same density operator $\rho_{B} = \frac{1}{2}I_{B}$. If a
quantum state prepared as an equal weight mixture of the states $| b_{1}
\rangle, | c \rangle, | d \rangle$ could be distinguished from a state
prepared as an equal weight mixture of the states $| b_{2} \rangle, | e
\rangle, | f \rangle$, the representation of quantum states by density
operators would be incomplete.

Now consider the effect of conditionalization on the state of A. The state
of AB can be expressed as the biorthogonal (Schmidt) decomposition: 
\begin{equation}
| \psi^{AB} \rangle = \frac{1}{\sqrt{2}} (| g \rangle| b_{1} \rangle + | h
\rangle| b_{2} \rangle)
\end{equation}
where 
\begin{eqnarray}
| g \rangle & = & \frac{2| a_{1} \rangle - | a_{2} \rangle -| a_{3} \rangle}{
\sqrt{6}} \\
| h \rangle & = & \frac{| a_{2} \rangle - | a_{3} \rangle}{\sqrt{2}}
\end{eqnarray}
The density operator $\rho_{A}$, obtained by tracing $| \psi^{AB} \rangle$
over B, is: 
\begin{equation}
\rho_{A} = \frac{1}{2}| g \rangle\langle g | + \frac{1}{2}| h \rangle\langle
h |
\end{equation}
which has support on a 2-dimensional subspace in the 3-dimensional Hilbert
space $\mbox{$\mathcal{H}$}_{A}$: the plane spanned by $| g \rangle$ and $|
h \rangle$ (in fact, $\rho_{A} = \frac{1}{2}P_{A}$, where $P_{A}$ is the
projection operator onto the plane). Conditionalizing on a value of $A$ or $
A^{\prime}$ yields a state that has a component outside this plane. So the
state change on conditionalization cannot be interpreted as an updating of
information in the sense of a refinement, i.e., as the selection of a
particular alternative among a set of mutually exclusive and collectively
exhaustive alternatives represented by the state $\rho_{A}$.

This is the notorious `irreducible and uncontrollable disturbance' arising
in the registration of new information about the occurrence of an event that
underlies the measurement problem: the loss of some of the information
encoded in the original state (in the above example, the probability of the
A-event represented by the projection operator onto the 2-dimensional
subspace $P_{A}$ is no longer 1, after the registration of the new
information about the observable $A$ or $A^{\prime}$). If the registration
of new information is the outcome of a measurement then, since the state
change on measurement will have to be stochastic and non-unitary, it cannot
be described by the deterministic dynamics of the theory, which must be
unitary (for closed systems) for consistency with the Hilbert space
representation of probabilities. A solution to the problem is generally
understood to require amending the theory in such a way that the loss of
information can be accounted for dynamically, and the quantum probabilities
can be reconstructed dynamically as measurement probabilities. Then the
quantum probabilities are not `uniquely given from the start' as kinematic features of
an appropriately represented event structure, i.e., they do not arise
kinematically but are derived dynamically, as artifacts of the measurement
process or of decoherence. Even on the Everett interpretation, where Hilbert
space is interpreted as the representation space for a new sort of
ontological entity, represented by the quantum state, and no definite
outcome out of a range of alternative outcomes is selected in a quantum measurement
process (so no explanation is required for such an event), probabilities
arise as a feature of the branching structure that emerges in the dynamical
process of decoherence.

From the perspective of the information-theoretic interpretation, the
`disturbance' involved in conditionalization is a kinematic phenomenon associated with the non-Boolean quantum event space. If there were no information loss in
the conditionalization of quantum probabilities, then cloning would be
possible, and equivalent mixtures associated with the same density operator
would be distinguishable, in which case Hilbert space would not be an
appropriate representation space for quantum events and their probabilistic
correlations.\footnote{
For the Everettian, there is the appearance of measurement disturbance on
each branch, or rather, on `most' branches, because there will always be
some branches on which it appears that there is no measurement
disturbance---and on these branches it will appear that cloning is possible.}
In the Appendix, we show that this follows directly from the `no cloning'
principle for a large class of theories. We prove that in this
class of theories the `no cloning' principle demarcates the boundary
between classical theories and theories in which measurement involves an `irreducible and
uncontrollable disturbance'. It seems plausible, therefore, that this
principle should play a central role in a derivation of the Hilbert space
structure from information theory.

It is instructive here to recall Einstein's distinction between `principle'
theories, like the special theory of relativity, formulated in terms of the
relativity principle and the light postulate (empirical regularities raised
to the level of postulates), and `constructive' theories, like Lorentz's
theory, formulated in terms of a rich ontology of objects like
particles, fields, and the ether. Einstein compared thermodynamics as a
principle theory (`no perpetual motion machines of the first and and second
kind') to the kinetic theory of gases as a constructive theory (where the
mechanical and thermal behavior of a gas is reduced to the motion of
molecules, modeled as little billiard balls). He proposed special relativity
as a kinematic replacement for Lorentz's dynamical interpretation of what we now refer to as Lorentz
covariance, which he saw as unsatisfactory, not as a rival theory of matter
and radiation. One might say that what eventually replaced Lorentz's theory
was relativistic quantum theory. From this perspective, Minkowski space-time
is the constructive theory corresponding to Einstein's principle theory
formulation of special relativity: it is a component of the kinematic part
of the constructive theory of the constitution of matter provided by
relativistic quantum theory. (See Janssen \cite[331--332]
{JanssenBalashov2003} for an account along these lines.)

In an article entitled `How to Teach Special Relativity' \cite
{BellRelativity}, John Bell considers the following puzzle: Three identical
spaceships, $A, B$, and $C$, are at rest relative to one other, drifting
freely far from other matter without rotation, with $A$ equidistant from $B$
and $C$. The spaceships $B$ and $C$ are connected by a fragile thread, which
is just long enough to span the distance between them. On reception of a
signal from $A$, the spaceships $B$ and $C$ start their engines and
accelerate gently. Since $B$ and $C $ are assumed to be identical, with
identical acceleration programs, they will have the same velocity and so
remain separated by the same distance relative to $A$. When $B$ and $C$
reach a certain velocity, the thread breaks. The question is: why does the
thread break? Note that the thread would not break under similar assumptions in a Newtonian universe.

The relativistic kinematic explanation goes along the following lines:

Let $F1$ be the inertial frame in which the spaceships $A, B, C$ are \emph{
initially} at rest (and $A$ remains at rest). In $F1$, the distance between $
B$ and $C$, as the spaceships begin to move and continue moving, remains the
same as the initial resting distance. But the moving thread undergoes a
Lorentz contraction in the direction of its motion in $F1$. The explanation,
in $F1$, of why the thread breaks is just this: the thread breaks because it
is contracting, and this contraction is resisted by the thread being tied to 
$B$ and $C$, which maintain a distance apart greater than the contraction
requires. The thread will break when $B$ and $C$ reach a sufficiently high
velocity in $F1$ and the prevention of the Lorentz contraction produces
sufficient stress to break the thread.

Let $F2$ be the inertial frame in which $B$ and $C$ are \emph{finally} at
rest again, after their engines have been shut off. From the perspective of $
F2$, there is a different explanation for the thread breaking. In $F2$, the
two spaceships $B$ and $C$ are decelerating, and eventually come to rest.
However, they are not decelerating at the same rate (they would be if $B$
and $C$ were connected by a rigid rod). It is this difference in
deceleration that is responsible for the stress in the thread, which
eventually causes the thread to break.

To clarify further, one might consider two additional spaceships, $E$ and 
$F$, identical to $B$ and $C$, with identical acceleration programs, initially
at rest in $F1$ (before $B$ and $C$ start their engines), with $E$ adjacent
to $B$, and $F$ adjacent to $C$. Suppose $E$ and $F$ are connected rigidly,
so that $EF$ behaves like a rigid rod with the two spaceships as endpoints,
initially at rest in $F1$. Suppose also that $EF$ starts accelerating at the
same time as $B$ and $C$ in $F1$, and that the rod connecting $E$ and $F$ is
strong enough to remain rigid under the acceleration. Bell's
characterization of the setup requires that, in $F1$, the distance between $
B $ and $C$, as the spaceships begin to move and continue moving, remains
the same as the initial resting distance. So, in $F1$, this distance will
become greater than the distance between $E$ and $F$, once the spaceships
start moving, since $EF$ will suffer a Lorentz contraction in the direction
of its motion. In the explanation in frame $F1$, the thread breaks because
it is contracting by as much as $EF$ contracts. In the explanation in frame $
F2$, $B$ and $C$ are not decelerating at the same rate---rather, the
endpoints of $EF$ are decelerating at the same rate---and this difference in
deceleration, relative to the deceleration of $EF$, is responsible for the
stress in the thread, which eventually causes it to break.

The explanations are frame-dependent, insofar as they involve elements that
are frame-dependent notions in special relativity. However, the increasing stress in
the thread that causes it to break, and the fact that the thread breaks when the stress exceeds the tensile strength of the thread, are
frame-independent features common to all explanations. What Bell pointed out
was that one ought to be able to provide an explanation for the thread
breaking in terms of an explicit calculation of the forces involved, and the
tensile strength of the thread. He suggests that such a dynamical
explanation is a deeper or at least more informative explanation than the
kinematic explanation. Harvey Brown's book \emph{Physical Relativity} \cite
{BrownBook} develops this theme.

In Bell's spaceship example, the dynamical explanation for the thread
breaking in terms of forces, insofar as the forces are Lorentz covariant,
shows the possibility of a dynamics consistent with the kinematics of
special relativity. The only factor relevant to the thread breaking is the
Lorentz contraction, a feature of the geometry of Minkowski space-time which
is quite independent of the material constitution of the thread and the
nature of the specific interactions involved. Given Einstein's two
principles, there is no deeper explanation for the thread breaking than the
kinematic explanation provided by the structure of Minkowski space-time.
\footnote{Harvey Brown's book \cite{BrownBook} presents an extended argument for the
contrary view.} The demonstration that a dynamical explanation yields the
same result as the kinematic explanation sketched above amounts to a
consistency proof that a relativistic dynamics---a dynamics that conforms to
the structure of Minkowski space-time---is possible.

If we take special relativity as a template for the analysis of quantum
conditionalization and the associated measurement problem,\footnote{
See Brown and Timpson \cite{BrownTimpson2007} for a contrary view.} the
information-theoretic view of quantum probabilities as `uniquely given from
the start' by the structure of Hilbert space as a kinematic framework for an indeterministic physics is the
proposal to interpret Hilbert space as a constructive theory of
information-theoretic structure or probabilistic structure, part of the kinematics of a full
constructive theory of the constitution of matter, where the corresponding
principle theory includes information-theoretic constraints such as `no
signaling' and `no cloning.'\footnote{
While the `no cloning' principle demarcates classical from non-classical
theories, we require some further principle or principles to recover Hilbert
space and exclude `superquantum' theories for which the correlation of
entangled states violates the Tsirelson bound for quantum states, while
conforming to the `no signaling' constraint. See Barnum \emph{et al} \cite
{BBLW2006,BBLW2007}.} Lorentz contraction is a physically real phenomenon
explained relativistically as a kinematic effect of motion in a
non-Newtonian space-time structure. Analogously, the change arising in
quantum conditionalization that involves a real loss of information is
explained quantum mechanically as a kinematic effect of \emph{any} process
of gaining information of the relevant sort in the non-Boolean probability structure
of Hilbert space (irrespective of the dynamical processes involved in the
measurement process). Given `no cloning' as a fundamental principle, there
can be no deeper explanation for the information loss on conditionalization
than that provided by the structure of Hilbert space as a probability theory
or information theory. The definite occurrence of a particular event is constrained
by the kinematic probabilistic correlations encoded in the
structure of Hilbert space, and only by these correlations---it is otherwise
`free.'

The Born weights are probabilities in a purely formal sense unless they are related to
experience by some explicitly formulated principle. The cash value of the `transition probability' $|\langle e|f\rangle|^{2}$ is that $|\langle e|f\rangle|^{2}$ represents the probability,
in the state $| e \rangle$, of finding the outcome corresponding to the
state $| f \rangle$ in a measurement of an observable of which $| f \rangle$
is an eigenstate. But if quantum mechanics is more than an instrument for
predicting the probabilities of measurement outcomes, it must be possible,
in principle, to locate structures that represent macroscopic measuring
instruments and recording devices in Hilbert space, where the dynamical
behavior of such structures is consistent with the kinematic information-theoretic
(probabilistic) principles encoded in the structure of Hilbert space.

In special relativity one has a consistency proof that a dynamical account
of relativistic phenomena in terms of forces, like the breaking of the
thread in Bell's spaceship example, is consistent with the kinematic
account in terms of the structure of Minkowski space-time. An analogous
consistency proof for quantum mechanics would be a dynamical explanation for
the effective emergence of classicality, i.e., Booleanity, at the
macrolevel, because it is with respect to the Boolean algebra of the
macroworld that the Born weights of quantum mechanics have empirical cash
value.

In classical mechanics, taking a Laplacian view, one can consider the phase space of the entire universe, in principle. The classical state, represented by a point in phase space that evolves dynamically, defines a 2-valued homomorphism on the Boolean algebra of (Borel) subsets of phase space, distinguishing events that occur at a particular time from events that don't occur. In this sense, the classical state is the `truthmaker' for propositions about the occurrence or non-occurrenc of events, for all possible events.

Similarly,  in quantum mechanics one can consider the Hilbert space of the entire universe, in principle. This is a space of possible events, with a certain kinematic structure of probabilistic correlations between events, represented by the subspace structure or projective geometry of the space (different from the classical correlational structure represented by the subset structure of phase space). On the usual view, the quantum analogue of the classical state is a pure state represented by a ray or 1-dimensional subspace in Hilbert space. There is, of course, no 2-valued homomorphism on the non-Boolean algebra of subspaces of Hilbert space, but a pure state can be taken as distinguishing events that occur at a particular time (events represented by subspaces containing the state, and assigned probability 1 by the state)  from events that don't occur (events represented by subspaces orthogonal to the state, and assigned probability 0 by the state). This leaves all remaining events represented by subspaces that neither contain the state nor are orthogonal to the state (i.e., events assigned a probability $p$ by the state, where $0 < p < 1$) in limbo: neither occurring nor not occurring. The measurement problem then arises as the problem of accounting for the fact that an event that neither occurs not does not occur when the system is in a given quantum state can somehow occur when the system undergoes a measurement interaction with a macroscopic measurement device---giving measurement a very special status in the theory. Once the pure state is taken as the analogue of the classical state in this sense, the only way out of this problem, without adding structure to the theory, is the Everettian manoeuvre.

On the information-theoretic interpretation, the quantum state is a derived entity, a credence function that assigns probabilities to events in alternative Boolean algebras associated with the outcomes of alternative measurement outcomes. The measurement outcomes are macro-events in a particular Boolean algebra, and the macro-events that actually occur, corresponding to a particular measurement outcome, define a 2-valued homomorphism on this Boolean algebra. What has to be shown is how this occurrence of events in a particular Boolean algebra is consistent with the quantum dynamics.

It is a contingent feature of the dynamics of our particular quantum universe that events represented by subspaces of Hilbert space have a tensor product structure that reflects the division of the universe into microsystems (e.g., atomic nuclei), macrosystems (e.g., macroscopic measurement devices constructed from pieces of metal and other hardware), and the environment (e.g., air molecules, electromagnetic radiation). The Hamiltonians characterizing the interactions between microsystems and macrosystems, and the interactions between macrosystems and their environment, are such that a certain relative structural stability emerges at the macrolevel as the  tensor-product structure of events in Hilbert space evolves under the unitary dynamics. Symbolically, an event represented by a 1-dimensional projection operator like $P_{\ket{\psi}} = \ket{\psi}\bra{\psi}$, where
\begin{equation}
\ket{\psi} = \ket{s}\ket{M}\ket{\varepsilon}
\end{equation}
and $s, M, \varepsilon$ represent respectively microsystem, macrosystem, and environment, evolves under the dynamics to $P_{\ket{\psi(t)}}$, where
\begin{equation}
\ket{\psi(t)} = \sum_{k}c_{k}\ket{s_{k}}\ket{M_{k}}\ket{\varepsilon_{k}(t)}, \label{eq:correlation}
\end{equation}
and
\begin{equation}
\ket{\varepsilon_{k}(t)} = \sum_{\nu}\gamma_{\nu}e^{-ig_{k\nu}t}\ket{e_{\nu}}
\end{equation}
if the interaction Hamiltonian $H_{M\varepsilon} $ between a macrosystem and the environment takes the form
\begin{equation}
H_{M\varepsilon} = \sum_{k\gamma}g_{k\nu}\ket{M_{k}}\bra{M_{k}}\otimes\ket{e_{\nu}}\bra{e_{\nu}}
\end{equation}
with the $\ket{M_{k}}$ and the $\ket{e_{k}}$ orthogonal.  That is,  the `pointer' observable $\sum_{k}m_{k}\ket{M_{k}}\bra{M_{k}}$ commutes with $H_{M\varepsilon}$ and so  is a constant of the motion induced by the Hamiltonian $H_{M\varepsilon}$. 

Here $P_{\ket{M_{k}}}$ can be taken as representing, in principle, a configuration of the entire macro\-world, and $P_{\ket{s_{k}}}$ a configuration of all the micro-events correlated with macro-events. The dynamics preserves the correlation represented by the superposition $\sum_{k}c_{k}\ket{s_{k}}\ket{M_{k}}\ket{\varepsilon_{k}(t)}$ between micro-events, macro-events, and the environment for the macro-events $P_{\ket{M_{k}}}$, even for non\-ortho\-gonal $\ket{s_{k}}$ and $\ket{\varepsilon_{k}}$, but not for macro-events $P_{\ket{M'_{l}}}$ where the $\ket{M'_{l}}$  are linear superpositions of the $\ket{M_{k}}$. Since the tri-decomposition $\sum_{k}c_{k}\ket{s_{k}}\ket{M_{k}}\ket{\varepsilon_{k}(t)}$ is unique (unlike the bi\-ortho\-gonal Schmidt decomposition; see Elby and Bub \cite{ElbyBub}), a correlation of the form $\ket{s}\ket{M}\ket{\varepsilon}$ evolves to a linear superposition in which the macro-events $P_{\ket{M'_{l}}}$ become correlated with entangled system-environment events represented by subspaces (rays) spanned by linear superpositions of the form $\sum_{k}c_{k}d_{lk}\ket{s_{k}}\ket{\varepsilon_{k}(t)}$.    (See Zurek \cite[p. 052105-14]{Zurek2005}.) 

It is characteristic of the dynamics that correlations represented by (\ref{eq:correlation}) evolve to similar correlations (similar in the sense of preserving the micro-macro-environment division), and the macro-events represented by $P_{\ket{M_{k}}}$, at a sufficient level of coarse-graining, can be associated with structures at the macrolevel---the familiar macro-objects of our experience---that remain relatively stable under the dynamical evolution. So a Boolean algebra $\hil{B_{M}}$ of macro-events $P_{\ket{M_{k}}}$ correlated with micro-events $P_{\ket{s_{k}}}$ in (\ref{eq:correlation}) is emergent in the dynamics. Note that the emergent Boolean algebra is not the same Boolean algebra from moment to moment, because the correlation between micro-events and macro-events changes under the dynamical evolution induced by the micro-macro interaction (e.g., corresponding to different measurement interactions).  What remains relatively stable under the dynamical evolution are the \textit{macrosystems} associated with macro-events in correlations of the form (\ref{eq:correlation}), even under a certain vagueness in the  coarse-graining associated with these macro-events: macrosystems like grains of sand, tables and chairs, macroscopic measurement devices, cats and people, galaxies, etc.

It is further characteristic of the dynamics that  the environmental events represented by $P_{\ket{\varepsilon_{k}(t)}}$
very rapidly approach orthogonality, i.e., the `decoherence factor'
\begin{equation}
\zeta_{kk'} = \langle\varepsilon_{k}|\varepsilon_{k'}\rangle = \sum_{\nu}|\gamma_{\nu}|^{2}e^{i(g_{k'\nu}-g_{k\nu})t}
\end{equation}
becomes negligibly small almost instantaneously.  When  the environmental events $P_{\ket{\varepsilon_{k}(t)}}$ correlated with the macro-events $P_{\ket{M_{k}}}$ are effectively orthogonal,  the reduced density operator is effectively diagonal in the `pointer' basis $\ket{M_{k}}$ and there is effectively no interference between elements of the emergent Boolean algebra $\hil{B_{M}}$. That is, the conditional probabilities of events associated with a subsequent emergent Boolean algebra (a subsequent measurement) are additive on $\hil{B_{M}}$.  (See Zurek \cite[p. 052105-14]{Zurek2005}, \cite{Zurek2003a}.)

The Born probabilities are probabilities of events in the emergent Boolean algebra, i.e.,  the Born probabilities are probabilities of `pointer' positions, the coarse-grained basis selected by the dynamics. Applying quantum mechanics kinematically, say in assigning probabilities to the possible outcomes of a measurement of some observable of a microsystem, we consider the Hilbert space of the relevant degrees of freedom of the microsystem and treat the measuring instrument as simply selecting a Boolean subalgebra in the non-Boolean event space of the microsystem to which the Born probabilities apply. In principle, we can include the measuring instrument in a dynamical analysis of the measurement process, but such a dynamical analysis---even though complete in terms of the quantum dynamics---does not provide a dynamical explanation of how individual outcomes come about. In such a dynamical analysis, the Born probabilities are probabilities of the occurrence of events in an emergent Boolean algebra. The information loss on conditionalization relative to classical conditionalization is a kinematic feature of the the structure of quantum events, not accounted for by the unitary quantum dynamics, which conforms to the kinematic structure. This is analogous to the situation in special
relativity, where Lorentz contraction is a kinematic effect of relative
motion that is \emph{consistent} with a dynamical account in terms of
Lorentz covariant forces, but is not explained in Einstein's theory---by
contrast with Lorentz's theory---as a dynamical effect in a Newtonian
space-time structure, in which this sort of contraction does not arise as a
purely kinematic effect. That is, the dynamical explanation of Lorentz
contraction in special relativity involves forces that are Lorentz
covariant---in effect, the dynamics is assumed to have symmetries that
respect Lorentz contraction as a kinematic effect of relative motion.
In quantum mechanics, the possibility of a dynamical analysis of the measurement process conforming to the kinematic structure of Hilbert space provides a consistency proof that the familiar objects of our macroworld behave dynamically in accordance with the kinematic probabilistic constraints on correlations between events.

A physical theory of an indeterministic universe is
primarily a theory of probability (or information). Probabilities are defined over an event structure, which in the quantum case is a family of Boolean algebras forming a particular sort of non-Boolean algebra. On the information-theoretic interpretation, no assumption is made about the fundamental `stuff' of the universe. So, one might ask, what do tigers supervene on?\footnote{We thank Allen Stairs for raising the realism question in this form.} In the case of Bohm's theory or the GRW theory, the answer is relatively straightforward: tigers supervene on particle configurations in the case of Bohm's theory, and on mass density or `flashes' in the case of the GRW theory, depending on whether one adopts the GRWm version or the GRWf version. In the Everett interpretation, tigers  supervene on features of the quantum state, which describes an ontological entity. In the case of the information-theoretic interpretation, the `supervenience base' is provided by the dynamical analysis: tigers supervene on events defining a 2-valued homomorphism in the emergent Boolean algebra. 

It might be supposed that this involves a contradiction. What is contradictory is to suppose that a correlational event represented by $P_{\ket{\psi(t)}}$ actually occurs, where
$\ket{\psi(t)}$ is a linear superposition $\sum_{k}c_{k}\ket{s_{k}}\ket{M_{k}}\ket{\varepsilon_{k}(t)}$, as well as an event represented by $P_{\ket{s_{k}}\ket{M_{k}}\ket{\varepsilon_{k}(t)}}$ for some specific $k$. We do not suppose this. On the information-theoretic interpretation we propose, there is a kinematic structure of possible correlations (but no particular atomic correlational event is selected as the `state' in a sense analogous to the pure classical state), and a particular dynamics that preserves certain sorts of correlations, i.e., correlational events of the sort represented by $P_{\ket{\psi(t)}}$ with $\ket{\psi(t)} = \sum_{k}c_{k}\ket{s_{k}}\ket{M_{k}}\ket{\varepsilon_{k}(t)}$ evolve to correlational events of the same form. What can be identified as emergent in this dynamics is an effectively classical probability space: a Boolean algebra with atomic correlational events of the sort represented by orthogonal 1-dimensional subspaces $P_{\ket{s_{k}}\ket{M_{k}}}$, where the probabilities are generated by the reduced density operator obtained by tracing over the environment, when the correlated environmental events are effectively orthogonal.

The dynamics does not describe the (deterministic or stochastic) evolution of the 2-valued homomorphism on which tigers supervene to a new 2-valued homomorphism (as in the evolution of a classical state). Rather, the dynamics leads to the relative stability of certain event structures at the macrolevel associated with the familiar macrosystems of our experience, and to an emergent effectively classical probability space whose atomic events are correlations between events associated with these macrosystems and micro-events. 

It is part of the information-theoretic interpretation that events defining a 2-valued homomorphism on the Boolean algebra of this classical probability space actually occur  with the emergence of the Boolean algebra at the macrolevel. This selection of actually occurring events is only in conflict with the quantum pure state if  the quantum pure state is assumed to have an ontological significance analogous to the ontological significance of the classical pure state as the `truthmaker' for propositions about the occurrence and non-occurrence of events, and if the quantum pure state evolves unitarily---in particular, if it is assumed  that the quantum pure state partitions all events into events that actually occur, events that do not occur, and events that neither occur nor do not occur, as on the usual interpretation. We argued that this assumption is one of the dogmas about quantum mechanics that should be rejected.  Rather, we take the quantum state, pure or mixed, to represent a credence function: the credence function of a rational agent (an information-gathering entity `in' the emergent Boolean algebra) who is updating probabilities on the basis of events that occur in the emergent Boolean algebra. 

\section{Concluding Remarks}

We have argued that the `big' measurement problem is like the problem for Newtonian physics raised by relativistic effects such as length contraction and time dilation, and that the solution to both problems involves the recognition of a fundamental change in the underlying \emph{kinematics} of our physics, represented by the transition from a Newtonian space-time to Minkowski space-time in the case of special relativity, and from the set-theoretic structure of classical phase space to the subspace structure of Hilbert space in the case of quantum mechanics. So the two assumptions, about the ontological significance of the quantum state and about the dynamical account of how measurement outcomes come about, should be rejected as unwarranted dogmas about quantum mechanics. 

The solutions to the `big' measurement problem provided by Bohm's theory and the GRW theory are dynamical and involve adding structure to quantum mechanics. There is a sense in which adding structure to the theory to solve the measurement problem dynamically---insofar as the problem arises from a failure to recognize the significance of Hilbert space as  the kinematic framework   for the physics of an indeterministic universe---is like Lorentz's attempt to explain relativistic length contraction dynamically, taking the Newtonian space-time structure as the underlying kinematics and   invoking      the ether as an additional structure for the propagation of electromagnetic effects.   In this sense, Bohm's theory and the GRW theory  are `Lorentzian' interpretations of quantum mechanics.  

The Everettian rejects the legitimacy of the problem by simply denying that measurements have definite outcomes, i.e., by denying that  the    pure states in a superposition describe alternative event complexes, only one of which actually occurs. This requires showing that a \emph{particular} decomposition of the quantum state corresponding to our experience has a preferred significance, and that weights can be assigned to the individual terms in the preferred superposition that have the significance of probabilities, even though no one definite event complex is selected as actually occurring in contrast to the other event complexes in the superposition. The Everettian's solution to \emph{this} problem is dynamical. So the Everettian, too, sees the underlying problem as dynamical. 

We reject the legitimacy of the `big' measurement problem on the basis of an information-theoretic interpretation of quantum mechanics, in terms of which the problem arises from the failure to see the significance of Hilbert space as the kinematic framework for an indeterministic physics. The dynamical analysis we provide is a solution to a consistency problem:  the `small' measurement problem. The analysis shows that a quantum dynamics, consistent with the kinematics of Hilbert space, suffices to underwrite the emergence of a classical probability space for the familiar macro-events of our experience, with the Born probabilities for macro-events associated with measurement outcomes derived from the quantum state as a credence function. The explanation for such nonclassical effects as the loss of information on conditionalization is not provided by the dynamics, but by the kinematics, and given `no cloning' as a fundamental principle, there can be no deeper explanation. In particular, there is no dynamical explanation for the definite occurrence of a particular measurement outcome, as opposed to other possible measurement outcomes in a quantum measurement process---the occurrence is constrained by the kinematic probabilistic correlations encoded in the projective geometry of Hilbert space, and only by these correlations. 

\section{Acknowledgements}

Jeffrey Bub acknowledges support from the National Science Foundation under Grant No. 0522398. Itamar Pitowsky's research is supported by the Israel Science Foundation, Grant 744/07.

\section{Appendix: The Information Loss Theorem}

We show that\emph{ it follows from the `no cloning' principle} that information
cannot be extracted from a nonclassical source without changing the source
irreversibly. (We prove this theorem for quantum information sources, but
note that the proof does not depend on specific features of the Hilbert
space formalism.)

We assume:

\begin{itemize}
\item[(1)] The `no cloning' principle: there is no universal cloning machine.

\item[(2)] Every (quantum) state $\rho$ is specified by the probabilities of
the measurement outcomes of a finite, informationally complete (or
`fiducial') set of observables.
\end{itemize}

Assumption (2) holds for a large class of theories, including quantum and
classical theories. Note that an informationally complete set is not unique.
For example, in the case of a qubit, the probabilities for spin `up' and
spin `down' in three orthogonal directions suffice to define a direction on
the Bloch sphere and hence to determine the state, so the spin observables $
\sigma_{x}, \sigma_{y}, \sigma_{z}$ form an informationally complete set.
(For a classical system or a classical information source, an
informationally complete set is given by of a single observable, with $n$
possible outcomes, for some $n$.)

Let $\mbox{$\mathcal{F}$} = \{A,B,C,\ldots\}$ be an informationally complete
set of observables represented by a finite set of Hermitian operators on an $
n$-dimensional Hilbert space $\mbox{$\mathcal{H}$}_{n}$. A quantum state $
\rho$ assigns a probability distribution to every outcome of any measurement
of an obervable in $\mbox{$\mathcal{F}$}$. Measuring $A$ yields one of the
outcomes $a_{1}, a_{2}, \ldots$ with a probability distribution $
P_{\rho}(a_{1}|A), P_{\rho}(a_{2}|A), \ldots$. Similarly, measuring $B$
yields one of the outcomes $b_{1}, b_{2}, \ldots$ with a probability
distribution $P_{\rho}(b_{1}|A), P_{\rho}(b_{2}|A), \ldots$, and so on. If $
\mbox{$\mathcal{F}$}$ is informationally complete, the finite set of
probabilities completely characterizes $\rho$ as the state on $
\mbox{$\mathcal{H}$}$.

Assuming that all measurement outcomes are independent and ignoring any
algebraic relations among elements of $\mbox{$\mathcal{F}$}$, a classical
probability measure on a classical (Kolmogorov) probability space can be
constructed from these probabilities: 
\begin{equation}
P_{\rho}(a,b,\ldots|A,B,\ldots) = P_{\rho}(a|A)P_{\rho}(b|B)\ldots
\end{equation}
(cf. the `trivial' hidden variable construction of Kochen and Specker in 
\cite{KochenSpecker}). Note that the probability space is finite since $
\mbox{$\mathcal{F}$}$ is finite and $\mbox{$\mathrm{dim}$}{
\mbox{$\mathcal{H}$}} < \infty$. (The number of atoms in the probability
space is at most $\mbox{$\mathrm{dim}$}{\mbox{$\mathcal{H}$}}^{|
\mbox{$\mathcal{F}$}|}$.) The quantum state $\rho$ can be reconstructed from 
$P_{\rho}$ (given as a classical information source, or rationally
approximated in the memory of a classical computer).

We now prove:

\begin{infolossthm}
Assumptions (1) and (2) entail that extracting information from a quantum information source given by a quantum state $\rho$, sufficient to generate the probabilities of an informationally
complete set of observables, is either impossible or necessarily changes the
state $\rho$ irreversibly, i.e., there must be information loss in the
extraction of such information.
\end{infolossthm}

\begin{proof}
Step 1: begin with a quantum source in the state $\rho$ and measure $A,B,\ldots$ sufficiently many times to generate the classical probability measure $P_{\rho}$, to as good an approximation as required, without  destroying $\rho$. Step 2: from $P_{\rho}$ construct a copy  of $\rho$. 
\begin{equation}
\rho \stackrel{\mbox{\scriptsize measure}}{\longrightarrow} P_{\rho} \stackrel{\mbox{\scriptsize prepare}}{\longrightarrow} \rho 
\end{equation}

This procedure defines a universal cloning machine, which we assume to be impossible. Since Step 2 is possible by assumption (2), the `no cloning' assumption (1) entails that Step 1 is blocked. 

We are left with two options: either there is no way to generate $P_{\rho}$ from $\rho$ (which is the case in quantum mechanics if we have only one copy of $\rho$, or too few copies of $\rho$), or else, if we can generate $P_{\rho}$ from $\rho$, assumption (1) entails that the original `blueprint' $\rho$ must have been changed irreversibly by the process of extracting the information to generate $P_{\rho}$ (if not, the change in $\rho$ could be reversed dynamically and cloning would be possible):
\begin{equation}
\xcancel{\rho} \stackrel{\mbox{\scriptsize measure}}{\longrightarrow} P_{\rho} \stackrel{\mbox{\scriptsize prepare}}{\longrightarrow} \rho \label{eqn:dist}
\end{equation}
\end{proof}

Since we can prepare multiple copies of the state $\rho$ from $P_{\rho}$,
one might think that even if the original state is destroyed in generating $
P_{\rho}$, we still end up with multiple copies of $\rho$: 
\begin{eqnarray}
& \xcancel{\rho} \overset{\mbox{\scriptsize measure}}{\longrightarrow} &
P_{\rho} \overset{\mbox{\scriptsize prepare}}{\longrightarrow} \rho  \notag
\\
& & \hspace{.2in} \overset{\mbox{\scriptsize prepare}}{\searrow} \rho  \notag
\\
& & \hspace{.35in} \vdots
\end{eqnarray}
But note that to generate $P_{\rho}$, we need to begin with multiple copies
of $\rho$, i.e., we need to begin with a state $\rho \otimes \rho \cdots$,
so what we really have is: 
\begin{equation}
\xcancel{\rho} \otimes \xcancel{\rho} \cdots \overset{
\mbox{\scriptsize
measure}}{\longrightarrow} P_{\rho} \overset{\mbox{\scriptsize prepare}}{
\longrightarrow} \rho \otimes \rho \cdots
\end{equation}
which simply re-states (\ref{eqn:dist}).

\begin{corollary}
No complete dynamical (i.e., unitary) account of the state transition in a
measurement process is possible in quantum mechanics, in general.
\end{corollary}

\begin{proof}
Any measurement can be part of an informationally complete set, so any measurement must lead to an irreversible (hence non-unitary) change in the quantum state of the measured system.
\end{proof}

We conclude---essentially from the `no cloning ' principle---that there can be no measurement device that functions
dynamically in such a way as to identify with certainty the output of an
arbitrary quantum information source without altering the source
irreversibly or `uncontrollably,' to use Bohr's term---no device can
distinguish a given output from every other possible output by undergoing a
dynamical (unitary) transformation that results in a state that represents a
distinguishable record of the output, without an irreversible transformation
of the source.

\bibliographystyle{plain}
\bibliography{Jeff}

\end{document}